\documentclass[12pt]{article}
\usepackage[left=1.5in, right=1.5in, top=1.5in, bottom=1.5in]{geometry}
\usepackage[onehalfspacing]{setspace}
\usepackage{amsmath,amsthm,amssymb,natbib,enumerate,xcolor,graphicx}
\usepackage[hyphens]{url}
\usepackage{hyperref}
\hypersetup{
    colorlinks,
    linkcolor={red!50!black},
    citecolor={blue!50!black},
    urlcolor={blue!50!black}
}
\newtheorem{thm}{Theorem}
\newtheorem{prop}[thm]{Proposition} 
\newtheorem{lemma}[thm]{Lemma}

\theoremstyle{definition}
\newtheorem{defi}{Definition}
\theoremstyle{remark}

\newcommand*{\set}[1]{\left\{#1\right\}}

\newcommand*{\x}{X}
\newcommand*{\y}{Y}

\newcommand*{\inswitchvcutknownP}{v_{k\rightarrow i}}

\title{Costlier switching strengthens competition even without advertising} 
\author{Sander Heinsalu\thanks{Research School of Economics, Australian National University, 25a Kingsley St, Acton ACT 2601, Australia.
Email: sander.heinsalu@anu.edu.au, 
website: \url{https://sanderheinsalu.com/}
No financial support was received for this research. 
}
}
\date{\today}

\begin{document}
\maketitle

\begin{abstract}
Consumers only discover at the first seller which product best fits their needs, then check its price online, then decide on buying. Switching sellers is costly.
Equilibrium prices fall in the switching cost, eventually to the monopoly level, despite the exit of lower-value consumers when changing sellers becomes costlier. 
More expensive switching makes some buyers exit the market, leaving fewer inframarginal buyers to the sellers. Marginal buyers may change in either direction, so for a range of parameters, all firms cut prices.

	
\end{abstract}


The purchasing process starts from some need the buyer has. For example, to heat the home, to cycle, cut wood, puree food or treat symptoms. The buyer then searches for products that may satisfy the need. Initially, the buyer does not know what products are available and which fits best. The buyer may go to a physical shop or search online, in which case the top search results are the websites of various sellers. The seller or the website explains the characteristics of the product and its differences from the offerings of competitors.\footnote{A pharmacy is legally obliged to explain the best drug or device, which may be sold by a competitor. Sellers of power tools, sports equipment and heating appliances minimise liability by explaining.
} 
For example, the difference between a geothermal heat pump and an air-source one, a gravel bike and a hybrid, a chainsaw and a disc-saw, a blender and a juicer, a CPAP machine and an oxygen concentrator. From this, the buyer learns her valuations for the products. The buyer then looks up the prices online (using her phone if in a physical shop). The buyer then chooses which seller, if any, to buy from, taking into account the \emph{switching cost} of going to a different seller from the current one. For physical shops, the switching cost is a transport cost. 
Online, the cost includes navigating another seller's site, creating an account and completing the order form. 

In such markets for infrequently purchased goods, a greater switching cost causes lower prices, as this paper shows. The lower prices reduce profit, which explains why similar businesses co-locate to reduce the switching cost. Examples are car dealerships on major roads out of town, fashion shops on the same street, entertainment venues in a city district, department stores in a mall. Industry associations, which aim to benefit their members, frequently publish a member directory\footnote{
Grocers: \url{http://www.agbr.com/store-locator/}, restoration contractors: \url{https://www.iicrc.org/page/IICRCGlobalLocator}, notaries: \url{https://www.thenotariessociety.org.uk/notary-search}.
} that reduces buyers' switching cost. 
Easier switching between sellers seemingly increases competition and reduces profits---an intuition confirmed by most industrial organisation models. Exceptions are discussed in the literature review below and rely on different economic forces from the current work. 

In this paper, differentiated sellers obtain higher profits by reducing the switching cost, because a higher cost makes more buyers exit, which reduces inframarginal demand for all sellers. The change in marginal demand depends on the distribution of buyer valuations. If this distribution is increasing, then the mass of marginal buyers increases, which unambiguously reduces prices. More generally, if the buyer valuation distribution does not fall too fast, then the inframarginal effect outweighs the change in marginal demand. 
The reduction in inframarginal demand is greater with more sellers, so prices are more likely to decrease in the switching cost in more competitive markets. 

In more detail, the environment is an oligopoly of horizontally differentiated sellers who simultaneously set prices. Initially, each consumer goes to a random seller. There, the consumer privately observes her valuations for all sellers and all their prices. Valuation may be determined by the brand and for a future service (flight, medical treatment, in-person training) the location and the service time. 
The consumer may buy at the initial seller, exit the market, or switch to a competitor. In the baseline model, sellers charge the same price to all buyers.\footnote{Section~\ref{sec:extensions} shows the results are robust when the sellers are able to distinguish their initial customers from those switching in from rivals.
For example, websites track buyers' browsing history to segment them into switchers and captive customers. 
} 

A greater switching cost causes some consumers to exit, because it becomes too costly to change sellers. Exit leaves less inframarginal demand in the market. If the sellers are similar enough, then all sellers obtain fewer inframarginal customers. If the valuation distribution is increasing, then the mass of marginal consumers rises in the switching cost. The optimal price decreases in the ratio of inframarginal to marginal buyers. 
Under mild conditions, the equilibrium is pure and unique. Prices are strategic complements, so the comparative statics are the same in all stable equilibria.

\textbf{Empirical evidence} 

Car sellers publish detailed descriptions of cars and prices online and also reduce consumers' switching costs by co-locating. Often, many car dealerships are adjacent next to a major road. 
In Table 11 of \cite{murry+zhou2020}, closing a co-located car seller reduces its next-door rivals' prices by 0.1\%. 

In a shopping mall, similar shops are adjacent: groceries on the ground floor, jewellery, cosmetics and clothes on higher floors. This reduces the cost of switching between them. The current paper proves that prices and profits increase as a result. 
Empirically, Table 8 of \cite{vitorino2012} shows that the profits of midscale department stores (e.g., Mervyn's, JC Penney) co-located in the same mall increase in their number.

\subsection*{Literature} 
The literature on price competition with switching costs between horizontally differentiated sellers mostly finds that prices increase in the cost. For example, in the Hotelling model with a strictly convex transport cost \citep{aspremont1979}. Articles in which prices decrease in the switching cost assume either multiperiod markets or add a countervailing force (e.g., advertising) to a higher switching cost. The present paper shows this effect in a simpler one-shot, one-product, no-advertising environment and identifies buyer exit as a novel cause for it. The exit effect is strong enough to make the price decrease strict for any positive switching cost at which some consumers still switch firms. 

The closest paper to the current one is \cite{heinsalu2020searchcost}, which also assumes that consumers discover their valuations for all goods at the first firm, but not the prices of other goods, unlike in this work. The results are similar, but the economic force is hold-up of switchers. Sellers would like to charge more than the switchers expected. The initial buyers leaving deters sellers from raising prices. The hold-up motive weakens at greater search cost, because fewer buyers search and switch.

\textbf{Multiperiod markets.}
\cite{klemperer1987} shows that if consumers in a multiperiod market face a switching cost after the first period, then firms compete in the first period to lock in customers. Higher switching costs intensify this competition and may cause below-cost prices in the first period. 
In \cite{choi+2018} and \cite{haan+2018}, firms advertise prices and compete to be the first that a consumer inspects in her ordered search. Prices fall in the consumer's cost of searching subsequent firms. The search cost effectively locks in the consumer. The current paper assumes firms do not advertise prices, which are unknown before the consumer visits the first seller. This shuts down the search direction effect and lock-in. Instead, the price decrease in the switching cost is due to exit, which changes the shape of demand (marginal and inframarginal consumers). 

In \cite{dube+2009,cabral2009,cabral2016}, in an infinite horizon model, a medium cost of switching causes lower prices than no cost. Reducing price to attract new customers turns out to be more important for the sellers than extracting more from existing customers with a high price. For switching costs high enough, prices still increase in the cost. 
The present paper focusses on a one-shot environment instead of investing in customer acquisition over time. Prices decrease in the switching cost throughout. 

\cite{cabral+villas-boas2005} state that in multiperiod competition when consumers are myopic, greater switching costs reduce profits and prices. 

\textbf{Countervailing forces.} 
Advertising to attract initial visits is widely studied. \cite{armstrong+zhou2011} examine a Hotelling duopoly---buyer valuations for the products are perfectly negatively correlated. A buyer learns her valuations at the first seller, as in the current work, but firms advertise prices or hire agents to make consumers initially direct their search to the firm. The current work features initially undirected search, then targeted buying. 

\cite{lal+sarvary1999} consider adding an online shop to a physical one. This decreases search costs but raises switching costs, because re-ordering a previous product is easy. For some parameter values, prices decrease in the switching cost, because the lower search cost dominates. In the current paper, the switching cost is the only force.

\textbf{Other related work.}
A large literature starting from \cite{diamond1971} examines search costs---consumers who visit a seller still have to pay to observe the valuations or prices of other products. \cite{wolinsky1986} is widely cited in this literature. The current work differs in that the first visit makes consumers perfectly informed about all prices and valuations. 

In \cite{mauring2021}, consumers learn prices sequentially at a cost, but valuations are commonly known. Firms may know a consumer's search cost or not, which may be positive or zero. Profit may rise in the number of firms, unlike in the current paper. Thus more competition may harm consumers. 

\cite{schultz2005} models a repeated Hotelling game, so valuations are perfectly negatively correlated. The larger the fraction of shoppers with zero search cost among buyers, the harder collusion is to sustain, so prices increase in the search cost, unlike in this work.

The next section sets up the model. Section~\ref{sec:results} derives the results. Section~\ref{sec:extensions} discusses extensions and implications, and concludes.

\section{The search for products followed by targeted buying}
\label{sec:model}

An oligopoly of $n$ firms indexed by $i$ simultaneously set prices $P_i$. The subscript $-i$ denotes firms other than $i$. 
There is a mass $1$ of consumers indexed by $v=(v_{i})$, where $v_{i}\in[0,1]$ is the consumer's valuation for firm $i$'s product. The valuation $v_i$ for firm $i$ is distributed according to $f_i$, which is positive with interval support. 
The corresponding cdf is denoted $F_i$. The valuations are independent across firms: $v_i\perp v_{-i}\;\forall i$. 


Fraction $\mu_i\in(0,1)$ of consumers initially arrive at firm $i$, independently of $v,P_i$. Clearly $\sum_{i=1}^{n}\mu_i=1$. Call the firm at which a consumer initially arrives the \emph{initial firm} of the consumer. After the consumer arrives, she observes her valuation for each firm and all prices (e.g., checks these online). Firms only have the common prior on the valuations. Each consumer decides whether to buy from her initial firm, change to the best other firm at cost $s>0$ and buy there, or exit. Because the consumer knows the prices and her valuations, she changes firms iff she buys at the other firm (never switches back to her initial firm or some other firm or exits after a change). 

The payoff from exiting is normalised to zero. 
A consumer with valuation $v$ who buys from her initial firm $i$ at price $P_i$ obtains payoff $v_i-P_i$. After changing firms, her payoff is $v_j-P_j-s$, where $j\neq i$. 

The marginal cost of firm $i$ is $c_i$. Firm $i$ that sets price $P_i$ resulting in \emph{ex post} demand $D_i$ gets \emph{ex post} profit $\pi_i:=(P_i-c_i)D_i$. 
Some results restrict attention to two firms, symmetric initial demands $\mu_i=\frac{1}{n}\;\forall i$ or symmetric marginal costs $c_i=0\;\forall i$, which should be apparent from the formulas. 

W.l.o.g.\ restrict $P_i\in\left[0,1\right]$, because a price that is negative or above the maximal valuation of consumers is never a unique best response. 
A mixed strategy of firm $i$ is the cdf $\sigma_i$ on $[0,1]$. 
Mixed prices on or off the equilibrium path are ruled out in Lemma~\ref{lem:pureknownP} below. 

Subgame perfect equilibrium consists of prices and consumer decisions such that (i) each firm maximises expected profit, given the decisions it expects from the rival firm(s) and the consumers. (ii) Consumers maximise their payoff by choosing to buy from their initial firm, change to buy from the other firm or exit. (iii) The expectations of the firms are correct. 

Denote the decision of a consumer with valuation $v$ who observes prices $P=(P_i)_{i=1}^{n}$ to buy from firm $j$ by $D(v,P)=j$ and to exit by $D(v,P)=0$. 
\begin{defi}
\label{def:equil}
Equilibrium is $P^* \in[0,1]^n$ and $D^*:[0,1]^{2n}\rightarrow \set{0,1,\ldots,n}$ such that for all $i$ and $v$, 
(i) $P_i^*\in\arg\max_{P_i}(P_i-c_i)D_i(P_i,P_{-i}^*)$, where $D_i(P) 
=\int_{\set{v:D^*(v,P)=i}}dF_1(v_1)\cdots dF_n(v_n)$, 
(ii) $D^*(v,P)=i\Rightarrow i\in\arg\max_{j}(v_j-P_j)$, (iii) $P_i=P_i^*$. 
\end{defi}

Asymmetric equilibria are allowed. With symmetric firms, Lemma~\ref{lem:stratcomplknownP} below provides sufficient conditions for all equilibria to be symmetric. 

The following section first derives the demand for a firm in terms of the valuation distributions and prices, then establishes the strategic complementarity of prices and uniqueness of equilibrium. The section concludes with the main result: the comparative statics of prices in the switching cost.

\section{The price response to the switching cost}
\label{sec:results}

A firm's demand consists of (1) consumers initially at that firm who buy immediately and (2) consumers initially at the rival firm who change firms. 
In a duopoly, the demand that firm $i$ expects from price $P_i$ when it expects firm $j$ to choose pricing strategy $\sigma_j^*$ is 
\begin{align}
\label{DknownP}
&\notag D_{i}(P_i,\sigma_j^*) 
=\mu_{i}\int_{c_j}^{1}\int_0^1\int_{P_i +\max\set{0,v_{j}-P_j^*-s}}^1f_i(v_{i})f_j(v_{j})dv_idv_jd\sigma_j^*(P_j^*) 
\\& +\mu_{j}\int_{c_j}^{1}\int_0^1\int_{P_i+s+\max\set{0,v_{j}-P_{j}^*}}^1f_i(v_{i})f_j(v_{j})dv_idv_jd\sigma_j^*(P_j^*).
\end{align} 
The outer integral in~(\ref{DknownP}) reflects firm $i$'s expectation over the prices of firm $j$. 

The next lemma establishes pure best responses of the firms in the two-firm case. 
Several alternative sufficient conditions for pure BR are available: either demand is log concave or each ccdf $1-F_i$ is concave or the densities of consumer valuations do not decrease too fast. Any weakly increasing density satisfies these, for example, uniform. Other examples satisfying these are truncated exponential, truncated normal, truncated Pareto with a negative location parameter and any $f_i$ that decreases slower than $\exp(-v_i^2 /4)$. 

The conditions in Lemma~\ref{lem:pureknownP} are far from necessary for pure equilibria. For example, if just one firm has a unique pure best response to any undominated strategy of the competitor, then generically the equilibrium is pure. 
\begin{lemma}
\label{lem:pureknownP}
If $D_i$ is log concave or $1-F_i(v_i)$ is concave or $(P_i-c_i)\frac{\partial f_i(P_i+w)}{\partial P_i}\geq - 2f_i(P_i+w)$ for all $P_i\in(c_i,1)$ and $P_i+w\in(c_i,1)$ for each firm $i$, then each has a pure best response to any $\sigma_j^*$.  
\end{lemma}
\begin{proof}[Proof of Lemma~\ref{lem:pureknownP}]
If $D_i$ in~(\ref{DknownP}) is log concave, then log profit is concave, thus firm $i$ has a pure BR. 

If $1-F_i(v_i)$ is concave, then its double integral~(\ref{DknownP}) is concave, because the sum of concave functions is concave and the integral is the limit of a sum. Concavity of $D_i$ implies log concavity. 

Using~(\ref{DknownP}), firm $i$'s first order condition (FOC) is
\begin{align}
\label{focknownP}
\frac{\partial \pi_i(P_i,\sigma_j^*)}{\partial P_i} 
&\notag=\int_{c_j}^{1}\int_0^1\left[1-\mu_iF_i\left(P_i +\max\set{0,v_{j}-P_j-s}\right)\right.
\\& -\mu_jF_i\left(P_i+s+\max\set{0,v_{j}-P_{j}^*}\right)
\\&\notag -(P_i-c_i)\mu_if_i\left(P_i +\max\set{0,v_{j}-P_j^*-s}\right) 
\\&\notag  \left.-(P_i-c_i)\mu_jf_i\left(P_i+s+\max\set{0, v_{j}-P_{j}^*}\right)\right]dF_j(v_j)d\sigma_j^*(P_j^*). 
\end{align}

If prices were unknown before changing firms, then the last line of~(\ref{focknownP}) would be multiplied by $\text{\textbf{1}}_{\set{P_i>P_{iCE}+s}}$, making~(\ref{focknownP}) larger at any prices, raising the best responses and equilibrium prices. 

The second order condition (SOC) with known prices is
\begin{align}
\label{socknownP}
\frac{\partial^2 \pi_i}{\partial P_i^2} 
&=-\left.\int_{c_j}^{1}\int_0^1 \right[ 2\mu_i f_i\left(P_i +\max\set{0,v_{j}-P_j^*-s}\right) 
\\&\notag +2\mu_j f_i\left(P_i+s +\max\set{0,v_j-P_j^*} \right)  
\\&\notag  +(P_i-c_i)\mu_i \frac{\partial f_i\left(P_i +\max\set{0,v_{j}-P_j^*-s}\right)}{\partial P_i}
\\&\notag \left. +(P_i-c_i)\mu_j \frac{\partial f_i\left(P_i +s +\max\set{0,v_{j}-P_j^*}\right)}{\partial P_i}\right]dF_j(v_j)d\sigma_j^*(P_j^*). 
\end{align} 


Sufficient for $\frac{\partial^2 \pi_i}{\partial P_i^2} <0\;\forall P_i\in(c_i,1)$ is $\int_0^1[(P_i-c_i)\frac{\partial f_i(P_i+w)}{\partial P_i} +2 f_i(P_i+w)]dF_j(v_j)\geq 0$ for all $P_i\in(c_i,1)$ and $P_i+w\in(c_i,1)$. 
This is ensured if $(P_i-c_i)\frac{\partial f_i(P_i+w)}{\partial P_i}\geq -2 f_i(P_i+w)$ for all  $P_i\in(c_i,1)$ and $P_i+w\in(c_i,1)$. 
Then the best response (BR) of firm $i$ to any $\sigma_j^*$ is pure and unique. 
\end{proof}

Assume pure prices from now on, which in a duopoly is w.l.o.g. 
Firm $i$ expects customers at competitor $k$ to switch to $i$ if 
$v_i-P_{i}-s >\max\{0,v_{k}-P_{k},\max_{j\neq i,k}\{v_{j}-P_{j}^*-s\}\}$, thus the $v_i$ cutoff above which consumers at firm $k$ switch to $i$ to be 
\begin{align}
\label{inswitchvcutknownP}
\inswitchvcutknownP:= P_{i} +\max\set{s,\; v_{k}-P_k^*+s,\; \max_{j\neq i,k}\set{v_{j}-P_{j}^*}}.
\end{align} 
A firm's price always affects the buying decision of consumers initially at other firms, unlike when buyers do not know prices before switching. 
The demand that firm $i$ expects is
\begin{align}
\label{Dpureoligopoly}
&\notag D_{i}(P_i,P_{-i}^*) =
\int_{[0,1]^{n-1}}\left\{\mu_i\left[1-F_i\left(P_i +\max\set{0,\max_{j\neq i}\set{v_{j}-P_{j}^*-s}}\right)\right]\right.
\\&\left.+\sum_{k\neq i}\mu_k\left[1-F_i\left(\inswitchvcutknownP\right)\right]\right\}dF_{-i}(v_{-i}). 
\end{align} 

The next lemma establishes the strategic complementarity of prices on and off the equilibrium path for any number of firms $n\geq2$. 
Several alternative sufficient conditions are available: either demand is log concave or the ccdf $1-F_i$ is concave or the densities of consumer valuations do not decrease too fast. Any weakly increasing density satisfies the latter, for example, uniform, triangular $f_i(v_i)=2v_i$, trapezoidal $f_i(v_i)=\min\set{kv_i,k-\sqrt{k^2-2k}}$ for $k>2$.
The truncated exponential and truncated normal distributions also satisfy the conditions, as does truncated Pareto with a negative location parameter. 

\begin{lemma}
\label{lem:stratcomplknownP}
If demand is log concave or $1-F_i$ concave or $f_i(P_i+w) +(P_i-c_i)\frac{\partial f_i(P_i+w)}{\partial P_i}\geq0$ for all $P_i\in(c_i,1)$ and $P_i+w\in(c_i,1)$, for each firm $i$, then prices are strategic complements, and 
decrease in $n$. 

If prices are strategic complements and the firms are symmetric, then any equilibrium is symmetric.
\end{lemma}
\begin{proof}
Using~(\ref{Dpureoligopoly}), firm $i$'s expected profit has the derivative 
\begin{align}
\label{focpureoligopolyknownP}
&\notag\frac{\partial \pi_i(P_i,P_{-i}^*)}{\partial P_i} 
=\int_{[0,1]^{n-1}}\left\{\mu_i\left[1-F_{i}\left(P_i +\max\set{0,\max_{j\neq i}\set{v_{j}-P_{j}^*}-s}\right)\right]\right.
\\&-\mu_i(P_i-c_i)f_i\left(P_i +\max\set{0,\max_{j\neq i}\set{v_{j}-P_{j}^*}-s}\right)
\\&\notag \left.+\sum_{k\neq i}\mu_k\left[1-F_{i}\left(\inswitchvcutknownP\right) -(P_i-c_i) f_i\left(\inswitchvcutknownP\right)\right]\right\}dF_{-i}(v_{-i}). 
\end{align}

Because $\inswitchvcutknownP \geq P_i +\max\set{0,\max_{j\neq i}\set{v_{j}-P_{j}^*}-s}$, if $f_i(P_i+w) +(P_i-c_i)\frac{\partial f_i(P_i+w)}{\partial P_i}\geq0$, then $F_{i}\left(\inswitchvcutknownP\right) +(P_i-c_i) f_i\left(\inswitchvcutknownP\right) \geq F_{i}\left(P_i +\max\set{0,\max_{j\neq i}\set{v_{j}-P_{j}^*}-s}\right) +(P_i-c_i)f_i\left(P_i +\max\set{0,\max_{j\neq i}\set{v_{j}-P_{j}^*}-s}\right)$ in~(\ref{focpureoligopolyknownP}). Increasing $n$ puts more weight on $1-F_{i}\left(\inswitchvcutknownP\right) -(P_i-c_i) f_i\left(\inswitchvcutknownP\right)$, which decreases~(\ref{focpureoligopolyknownP}) and the best response price. 

The second derivative of $i$'s expected profit is
\begin{align}
\label{socpureoligopolyknownP}
&\notag\frac{\partial^2 \pi_i(P_i,P_{-i}^*)}{\partial P_i^2} 
=\int_{[0,1]^{n-1}}\left[-2\mu_if_{i}\left(P_i +\max\set{0,\max_{j\neq i}\set{v_{j}-P_{j}^*}-s}\right)\right.
\\& -\mu_i(P_i-c_i)f_i'\left(P_i +\max\set{0,\max_{j\neq i}\set{v_{j}-P_{j}^*}-s}\right)
\\&\notag -\sum_{k\neq i}2\mu_kf_{i}\left(\inswitchvcutknownP\right) -\left.\sum_{k\neq i}\mu_k(P_i-c_i) f_i'\left(\inswitchvcutknownP\right)\right]dF_{-i}(v_{-i}). 
\end{align}

The game is supermodular if $\frac{\partial^2 \pi_i}{\partial P_i\partial P_j^*}\geq0\;\forall i,j$, in which case prices are strategic complements.  

The cross-partial derivative is 
\begin{align}
\label{crosspartialpureoligopolyknownP}
&\notag\frac{\partial^2 \pi_i(P_i,P_{-i}^*)}{\partial P_i\partial P_k^*} 
=\int_{[0,1]^{n-1}}\left[\text{\textbf{1}}_{\set{v_{k}-P_{k}^* =\max_{j\neq i}\set{v_{j}-P_{j}^*}\geq 0}}2\mu_if_{i}\left(P_i +\max\set{0,v_{k}-P_{k}^*-s}\right)\right.
\\&\notag 
 +\text{\textbf{1}}_{\set{v_{k}-P_{k}^* =\max_{j\neq i}\set{v_{j}-P_{j}^*}\geq 0}}\mu_i(P_i-c_i)f_i'\left(P_i +\max\set{0,v_{k}-P_{k}^*-s}\right)
\\& \left. +\mu_k\text{\textbf{1}}_{\set{v_{k}\geq P_{k}^*}} f_{i}\left(\inswitchvcutknownP\right) +\mu_k(P_i-c_i)\text{\textbf{1}}_{\set{v_{k}\geq P_{k}^*}} f_i'\left(\inswitchvcutknownP\right)\right]dF_{-i}(v_{-i}). 
\end{align}
A sufficient condition for strategic complementarity for any $n\geq2$ is $(P_i-c_i)\frac{\partial f_i(P_i+w)}{\partial P_i}\geq -f_i(P_i+w)$ for all $P_i\in(c_i,1)$ and $P_i+w\in(c_i,1)$. 

In an asymmetric equilibrium, the prices of at least two firms (labelled $1$, $2$ w.l.o.g.) are distinct. Any asymmetric equilibrium in a symmetric game remains an equilibrium after any permutation of player labels, so if an asymmetric equilibrium exists, then there exist at least two. 
If the firms are symmetric, then the best response of firm $i$ to $P_{-i}$ only depends on $\set{P_{j}:j\neq i}$, not the order of $P_j$, $j\neq i$ in $P_{-i}$. 

Suppose there exist asymmetric equilibria with price vectors $(P_i^*)\neq (P_i^{**})$ that are permutations of each other, i.e.\ $\set{P_i^*} =\set{P_j^{**}}$. W.l.o.g.\ label the firms so that 
$P_1^* =\min_jP_j^* <P_1^{**}$. 
Then $\set{P_j^{**}:j>1} =\set{P_j^*:j>1}\cup \set{\min_jP_j^*}\setminus \set{P_1^{**}}$, i.e.\ all elements but one coincide in $\set{P_j^{**}:j>1}$ and $\set{P_j^*:j>1}$, and the remaining element is strictly smaller in $\set{P_j^{**}:j>1}$. 
Strategic complements then imply that the best response to $\set{P_j^{**}:j>1}$ is less than the best response to $\set{P_j^*:j>1}$, contradicting $P_1^{**} >\min_jP_j^{**} =P_1^*$. 
\end{proof}

The conditions in Lemma~\ref{lem:stratcomplknownP} are sufficient but not necessary for prices to be strategic complements. 

The next lemma shows that in a duopoly, the equilibrium is unique if the consumer valuation pdf does not jump up (implied by continuity) and satisfies the conditions for a pure best response in Lemma~\ref{lem:pureknownP}. 
\begin{lemma}
\label{lem:uniqueknownP}
If $n=2$, $f_i$ has no upward jumps and either $D_i$ is log concave or $1-F_i(v_i)$ is concave or $(P_i-c_i)\frac{\partial f_i(P_i+w)}{\partial P_i}\geq - 2f_i(P_i+w)$ for all $P_i\in(c_i,1)$ and $P_i+w\in(c_i,1)$ for each firm $i$, then the equilibrium is unique. 
\end{lemma}
\begin{proof}
Profit after imposing the equilibrium condition $P_i =P_i^*$ on both firms is denoted $\pi_i^*$. 
Sufficient for uniqueness is that the best responses do not jump up and the slopes of best responses are below $1$ at any equilibrium, i.e.,  
$\left|-\frac{\partial^2 \pi_i^*}{\partial P_i\partial P_j}\left/\frac{\partial^2 \pi_i^*}{\partial P_i^2}\right.\right| <1$ for each firm $i\neq j$. 
Equivalently, 
$\frac{\partial^2 \pi_i^*}{\partial P_i\partial P_j} +\frac{\partial^2 \pi_i^*}{\partial P_i^2}<0$. 
If the slope of a function with no upward jumps is $<1$ at any fixed point, then there is a unique fixed point. 

From~(\ref{focpureoligopolyknownP}) and~(\ref{crosspartialpureoligopolyknownP}) for $n=2$, 
\begin{align*}
&\frac{\partial^2 \pi_i^*}{\partial P_i\partial P_j} +\frac{\partial^2 \pi_i^*}{\partial P_i^2} 
=\int_{P_j+s}^1 \mu_i\left[f_i\left(P_i +v_{j}-P_j-s\right) +(P_i-c_i)\frac{\partial f_i\left(P_i +v_{j}-P_j-s\right)}{\partial P_i}\right]dF(v_j)
\\&\notag+\int_{P_j}^1 \mu_j\left[f_i\left(P_i+v_{j}-P_{j}+s\right) +(P_i-c_i)\frac{\partial f_i\left(P_i +v_{j}-P_j+s\right)}{\partial P_i}\right]dF(v_j)
\\&\notag 
+\int_0^1\mu_i\left[-2f_i\left(P_i +\max\set{0, v_{j}-P_j-s}\right)
-(P_i-c_i)\frac{\partial f_i\left(P_i +\max\set{0, v_{j}-P_{j}-s}\right)}{\partial P_i}\right]dF(v_j)
\\&\notag +\int_0^1 \mu_j\left[-2f_i\left(P_i +s +\max\set{0, v_{j}-P_{j}}\right) -(P_i-c_i)\frac{\partial f_i\left(P_i+s +\max\set{0, v_{j}-P_{j}}\right)}{\partial P_i}\right]dF(v_j). 
\end{align*}
The effect of the SOC dominates: the integral is from $0$ to $1$ instead of from $P_j+s$ or $P_j$ to $1$, and the $f_i$ term (with a clear sign) is multiplied by $2$. Therefore if the SOC holds, then $\frac{\partial^2 \pi_i^*}{\partial P_i\partial P_j} +\frac{\partial^2 \pi_i^*}{\partial P_i^2} <0$. 
\end{proof}

The conditions in Lemma~\ref{lem:uniqueknownP} are sufficient, but not necessary for uniqueness. 
Further, uniqueness is not necessary for the main result because strategic complementarity makes the direction of comparative statics in all stable equilibria the same. 

With any number of firms, sufficient for uniqueness is $(P_i-c_i)\frac{\partial f_i(P_i+w)}{\partial P_i}\leq -f_i(P_i+w)$ for all $P_i\in(c_i,1)$ and $P_i+w\in(c_i,1)$, which guarantees the dominant diagonal condition $\sum_j\frac{\partial^2 \pi_i^*}{\partial P_i\partial P_j} <0$, generalising Lemma~\ref{lem:uniqueknownP}. The condition is the opposite to the sufficient condition for strategic complementarity in Lemma~\ref{lem:stratcomplknownP}.

The main theorem establishes that if the consumer valuation pdf is weakly increasing and convex, then each firm's price decreases in the search cost of the consumers.  
Uniform valuations satisfy the condition, as do trapezoidal increasing valuation distributions. 
The proof uses the Implicit Function Theorem. 
\begin{thm}
\label{thm:mainknownP}
If $\mu_i\approx \mu_{\ell}\;\forall i,\ell$ and either $f_{i}\left(x\right) +P_if_i'\left(x\right) \geq0$ and $f_{i}'\left(x\right) +P_if_i''\left(x\right) \geq0$ for all $P_i\in(c_i,1)$ and $x\in[P_i,1)$ or $f_i',f_i'' \geq 0$, then $\frac{dP_i^*}{ds} \leq 0$ for all firms in any stable equilibrium. If further $s<1-P_i^*\;\forall i$ and the equilibrium is symmetric, then $\frac{dP_i^*}{ds} <0$. 
\end{thm}
\begin{proof}
Impose the equilibrium condition $P_i=P_i^*$ in the FOC~(\ref{focpureoligopolyknownP}), 
take its derivative w.r.t.\ $s$ and split the integral into regions of $v$ based on which $\ell$ is $i$'s strongest competitor $\ell =\arg\max_{j\neq i}\set{v_{j}-P_{j}}$: 
\begin{align}
\label{dfocdsoligopolyknownP}
&\notag\frac{\partial FOC_i^*}{\partial s} 
=\sum_{\ell\neq i}\int_{{\set{v_{-i}:\ell=\arg\max_{j\neq i}\set{v_{j}-P_{j}}}}}\{\text{\textbf{1}}_{\set{v_{\ell}-P_{\ell}\geq s}}\mu_i[f_{i}\left(P_i +v_{\ell}-P_{\ell}-s\right) 
\\& +P_if_{i}'\left(P_i +v_{\ell}-P_{\ell}-s\right)]
\\&\notag -\mu_{\ell}\left[f_{i}\left(P_i+s+\max\set{0, v_{\ell}-P_{\ell}}\right) +P_if_i'\left(P_i+s+\max\set{0, v_{\ell}-P_{\ell}}\right)\right]
\\&\notag -\sum_{k\neq i,\ell}\mu_k\text{\textbf{1}}_{\set{v_{\ell}-P_{\ell} \leq\max\set{s,v_{k}-P_{k}+s}}}[f_{i}\left(P_i+s+\max\set{0, v_{k}-P_{k}}\right) 
\\&\notag +P_if_i'\left(P_i+s+\max\set{0, v_{k}-P_{k}}\right)]\}dF_{-i}(v_{-i}). 
\end{align}

Sufficient for $\frac{\partial FOC_i^*}{\partial s} <0$ is $\mu_i\approx \mu_{\ell}\;\forall i,\ell$ and $f_{i}\left(x+s\right) +P_if_i'\left(x+s\right) \geq0$ and 
$f_{i}'\left(z\right) +P_if_i''\left(z\right) \geq0$ for all $P_i\in(c_i,1)$, $x\in[P_i,1-s)$, $z\in[P_i,1)$, with at least one inequality strict. For this, $f_i',f_i'' \geq 0$ is sufficient. 
By the Implicit Function Theorem, 
\\$\left[\begin{array}{l} \frac{d P_{1}^*}{ds} \\ \ldots \\
\frac{d P_{n}^*}{ds}
\end{array}\right] =-\left[\begin{array}{l}
\frac{\partial FOC_{i}^*}{\partial P_{j}}
\end{array}\right]^{-1}\left[\begin{array}{l} \frac{\partial FOC_{1}^*}{\partial s} \\ \ldots \\
\frac{\partial FOC_{n}^*}{\partial s}
\end{array}\right]$. 
The matrix $\left[\frac{\partial FOC_{i}^*}{\partial P_{j}}\right]^{-1}$ is negative semidefinite iff the equilibrium is stable. 
Therefore sufficient for $\frac{dP_i^*}{ds}<0\;\forall i$ in any stable equilibrium is $\frac{\partial FOC_i^*}{\partial s} <0\;\forall i$. 

If no consumers learn, then $s$ does not affect prices, so $\frac{\partial FOC_i^*}{\partial s} =0\leq 0$. The condition $1-P_i^*-s>-P_j^*$ ensures that some consumers at $j$ learn about $i$, given the equilibrium prices. In a symmetric equilibrium $P_i^*=P_j^*$, so at any firm, some consumers learn if $s<1$. 
\end{proof}

Comparing~(\ref{dfocdsoligopolyknownP}) to the case of unknown prices, 
the terms $P_if_i'(P_i+s+\max\set{0, v_{m}-P_{m}})$ for $m=k,\ell$ are added. These may be positive or negative, so for either sign of $\frac{\partial FOC_i^*}{\partial s}$ under unknown prices, $\frac{\partial FOC_i^*}{\partial s}$ may have either sign under known prices. 

To interpret~(\ref{dfocdsoligopolyknownP}) in Theorem~\ref{thm:mainknownP}, suppose the strongest rival to firm $i$ is $\ell$. Then the term under the integral multiplied by $\text{\textbf{1}}_{\set{v_{\ell}-P_{\ell}\geq s}}$ describes the consumers initially at $i$ who stay at $i$ at a higher $s$ but would have left at a lower. The second term without an indicator function captures the customers at the strongest rival who would have switched to $i$ but at a higher $s$ do not. The same effect at less strong rivals is described by the last sum containing $\text{\textbf{1}}_{\set{v_{\ell}-P_{\ell} \leq\max\set{s,v_{k}-P_{k}+s}}}$, which is absent when there are only two firms. 
In each term, $f_i$ measures the decrease in inframarginal consumers (the $\frac{1}{n}-\frac{1}{n}F_i$ terms) in the FOC~(\ref{focpureoligopolyknownP}) when the argument of $F_i$ increases, and $f_i'$ measures the increase in marginal buyers, captured by the $f_i$ terms in~(\ref{focpureoligopolyknownP}). The argument of $F_i$ and $f_i$ decreases in $s$ for the customers initially at $i$, but for other consumers increases in $s$. For example, a greater $s$ raises the mass of inframarginal consumers staying at $i$, reduces the mass switching in, and if $f_i'>0$, then reduces the mass of marginal consumers among those initially at $i$ and increases this mass among those switching to $i$. If $f_i'<0$, then the effects on the marginal consumers are the opposite. 

If initial demands are similar and there are three or more firms, then the incoming switchers form most of a firm's demand and have a larger effect on its profit than its initial customers. This is also true with two firms because the search cost has a greater effect on the number of incoming switchers than the initial customers. Figure~\ref{fig:srisesknownP} illustrates how the switchers respond to the search cost on two margins: exit and which firm to buy from, but the initial customers respond just with the choice of firm. 

\begin{figure}[h!]
\caption{Demands after an increase in the search cost $s$ from $0.1$ to $0.2$ at $P_{\x}=0.6$, $P_{\y}=0.45$. Thick lines: marginal consumers.}
\label{fig:srisesknownP}
\includegraphics[width=\linewidth]{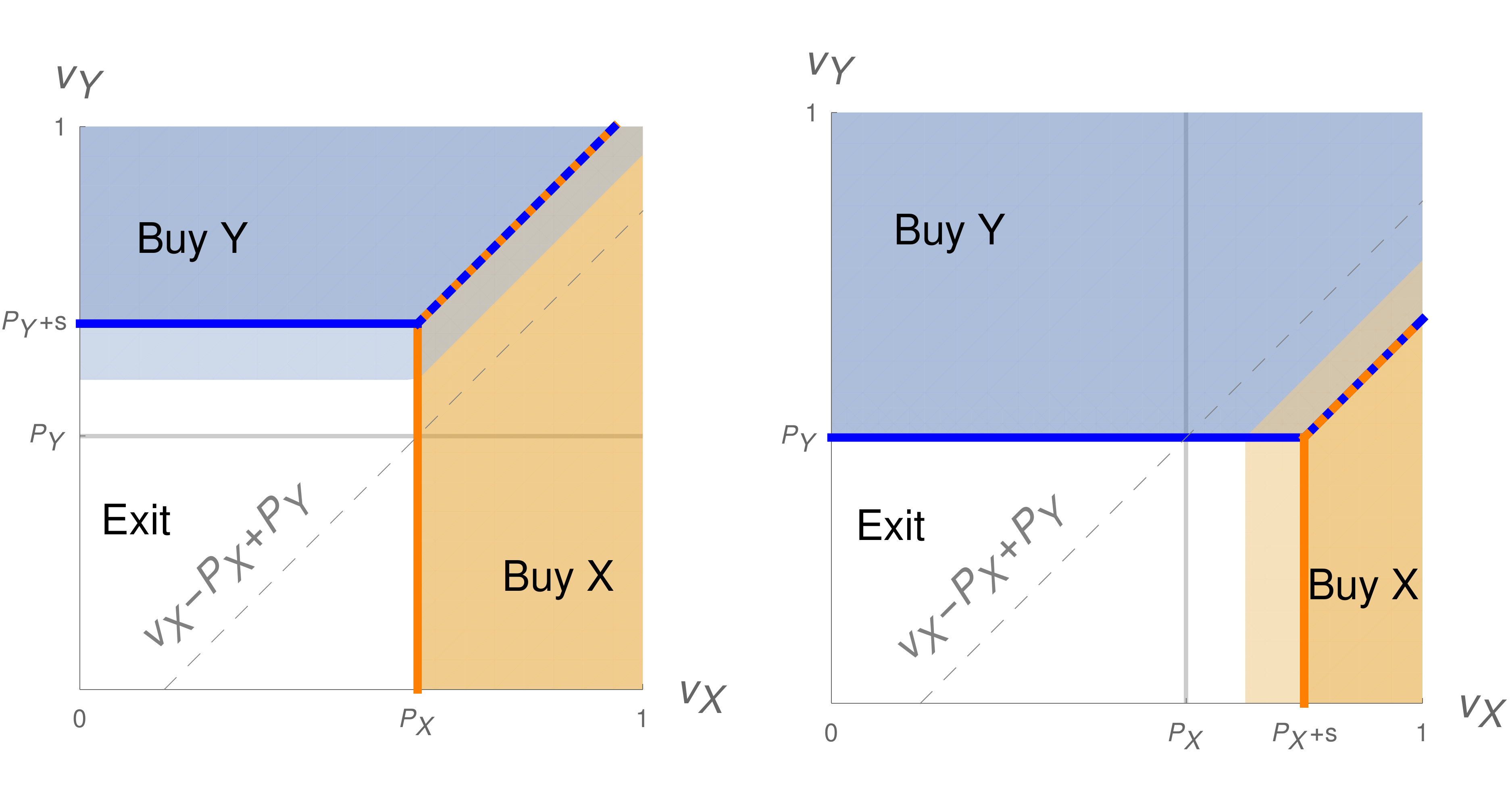} 
\end{figure}

The left panel of Figure~\ref{fig:srisesknownP} depicts demands for the two firms $\x$ and $\y$ from consumers initially at firm $\x$ and the right panel from those initially at $\y$. The thick blue line, including the dashed part, is the marginal consumers for firm $\y$ after the search cost rises, and the thick orange line, including dashed, the marginals for firm $\x$. The marginal consumers before the search cost increase are on the boundaries between the blue, orange and white areas. 

The change in inframarginal consumers for a given firm is clearly larger among the incoming switchers (in the other firm's panel) than among the initial customers (in the firm's own panel). For example, the initial customer change for firm $\x$ in the left panel only occurs where $v_{\y}\geq P_{\y}+s$ (reflected in~(\ref{dfocdsoligopolyknownP}) by multiplying with $\text{\textbf{1}}_{\set{v_{\ell}-P_{\ell}\geq s}}$), but in the right panel for $v_{\y}$ both above and below $P_{\y}$ (the last three lines in~(\ref{dfocdsoligopolyknownP})). The length of the marginal consumer line changes in opposite directions in the two panels, so the changes offset each other. If the pdf $f_i$ is increasing, then a shift of the marginal line away from the origin increases the mass of consumers on it. 
More marginal and fewer inframarginal consumers for each firm reduces its best response to any prices of the rivals. Thus equilibrium prices fall in the search cost if $f_i$ is increasing. 

In more detail, if $s$ increases by $\Delta s$ and the initial prices and valuation pdf-s are symmetric, then the mass of marginal consumers for firm $\y$ falls by $\mu_{\x}\sqrt{2}\Delta s$ in the left panel of  Figure~\ref{fig:srisesknownP} and by $\mu_{\y}(\sqrt{2}-1)\Delta s$ in the right. The absolute elasticity of marginal consumers is less than $[\mu_{\x}\sqrt{2} +\mu_{\y}(\sqrt{2}-1)]/1$. 
The inframarginal consumers for firm $\y$ decrease by at least $\mu_{\x}\Delta s$ and increase by at most $\mu_{\y}(1-P_{\x}-s)\sqrt{2}\Delta s$. If the equilibrium prices are above $\frac{1}{2}$, as in all numerical examples, then the absolute elasticity of inframarginal consumers for firm $\y$ is at least $[\mu_{\x}-\mu_{\y}(\frac{1}{2}-s)\sqrt{2}]/\frac{3}{8}$. With symmetric initial demands, the elasticity of marginal consumers is smaller than that of inframarginals. 

Suppose exit is not an option for the buyers and the firms are symmetric. Then after the search cost increases, the initial consumers who stay at a firm exactly offset the switchers who stop coming. The mass of inframarginal consumers remains constant. The marginal consumer line is parallel to the 45-degree line (no vertical component for firm $\x$ or horizontal for $\y$ in Figure~\ref{fig:srisesknownP}). In response to a higher search cost, the marginal line shifts parallel to itself away from the centre in each panel of Figure~\ref{fig:srisesknownP}. The marginal line thus shortens. If buyer valuations for the firms are independent or positively correlated and not too convex, then the mass of marginal consumers decreases. 
Therefore prices increase in the search cost---the opposite to what happens when buyers can exit. 
With an exit option, the inframarginal effect appears, which is a force toward prices decreasing in the search cost. The marginal effect under independent valuations starts to depend on the shape of the pdf of valuations. The effect of correlated valuations on the mass of marginal consumers is still present. 



Returning to the model where buyer exit is possible, if the search cost becomes so large that no consumers switch, then each firm's price falls to its monopoly level. 
This monopoly price is with respect to the remaining demand when low-valuation customers have exited and the rest buy from their initial firm even if their valuation is greater for a different firm---an inefficient allocation. 
Demand is smaller than at a lower switching cost, contains relatively more high-valuation customers overall, but each firm obtains relatively more buyers with a medium valuation for it and relatively fewer high-valuation buyers. 

The next result establishes that in a competitive enough market (in the limit of many similar firms), prices always decrease in the search cost. The reason is that most of the inframarginal consumers arrive from a firm's rivals, and many such buyers exit in response to a higher search cost. The mass of marginal consumers changes little. 
\begin{prop}
\label{prop:manyfirmsknownP}
If $(P_i-c_i)\frac{\partial f_i(P_i+w)}{\partial P_i}\geq -f_i(P_i+w)$ for all $P_i\in(c_i,1)$ and $P_i+w\in(c_i,1)$, then there exists $\underline{n}\in\mathbb{N}$ s.t.\ for any $n>\underline{n}$ symmetric firms, $\frac{dP_i^*}{ds} \leq 0$ for all firms in any stable equilibrium. If further $s<1-P_i^*$, then $\frac{dP_i^*}{ds} <0$. 
\end{prop}
\begin{proof}
Denote by $-i,\ell$ all the firms except $i$ and $\ell$. Rewrite~(\ref{dfocdsoligopolyknownP}) as 
\begin{align}
\label{manyfirms}
&\notag\frac{\partial FOC_i^*}{\partial s} 
=\sum_{\ell\neq i}\int_0^1\{\text{\textbf{1}}_{\set{v_{\ell}-P_{\ell}\geq s}}\mu_i[f_{i}\left(P_i +v_{\ell}-P_{\ell}-s\right) 
+(P_i-c_i)f_{i}'\left(P_i +v_{\ell}-P_{\ell}-s\right)]
\\&\notag -\mu_{\ell}\left[f_{i}\left(P_i+s+\max\set{0, v_{\ell}-P_{\ell}}\right) +(P_i-c_i)f_i'\left(P_i+s+\max\set{0, v_{\ell}-P_{\ell}}\right)\right]
\\&\notag -\sum_{k\neq i,\ell}\mu_k\int_{\set{v_{-i,\ell}:\ell=\arg\max_{j\neq i}\set{v_{j}-P_{j}}}}\text{\textbf{1}}_{\set{v_{\ell}-P_{\ell} \leq\max\set{s,v_{k}-P_{k}+s}}}[f_{i}\left(P_i+s+\max\set{0, v_{k}-P_{k}}\right) 
\\& +(P_i-c_i)f_i'\left(P_i+s+\max\set{0, v_{k}-P_{k}}\right)]dF_{-i,\ell}(v_{-i,\ell})\}dF_{\ell}(v_{\ell}). 
\end{align}
The $\sum_{k\neq i,\ell}$ term in~(\ref{manyfirms}) simplifies to $\sum_{k\neq i,\ell}\mu_k\int_{0}^1\text{\textbf{1}}_{\set{v_{\ell}-P_{\ell} \leq\max\set{s,v_{k}-P_{k}+s}}}[f_{i}(P_i+s+\max\set{0, v_{k}-P_{k}}) +(P_i-c_i)f_i'(P_i+s+\max\set{0, v_{k}-P_{k}})]dF_{k}(v_{k})$. 
For any $n,s$, bound $\sum_{k\neq i,\ell}$ in~(\ref{manyfirms}) below by 
$
\sum_{k\neq i,\ell}\mu_k\int_{P_k+v_{\ell}-P_{\ell}-s}^1\text{\textbf{1}}_{\set{ v_{k}\geq P_{k}}}[f_{i}\left(P_i+s+v_{k}-P_{k}\right) +(P_i-c_i)f_i'\left(P_i+s+v_{k}-P_{k}\right)]dF_{k}(v_{k}).
$
The bracketed term is positive if $(P_i-c_i)\frac{\partial f_i(P_i+w)}{\partial P_i}\geq -f_i(P_i+w)$. 

For any $P_k<1$, as $n\rightarrow\infty$, arbitrarily many $k$ satisfy $v_{k}> P_{k}$. If the equilibrium prices are similar enough ($P_k\approx P_{\ell}$, ensured by symmetric firms), then as $n\rightarrow\infty$, arbitrarily many $k$ satisfy $P_k+v_{\ell}-P_{\ell}-s <1$, in which case the inner integral in~(\ref{manyfirms}) is positive. Therefore the limit of $\sum_{k\neq i,\ell}$ as $n\rightarrow\infty$ is infinite and~(\ref{manyfirms}) is negative. Using the Implicit Function Theorem, as in Theorem~\ref{thm:mainknownP}, completes the proof. 
\end{proof}

Proposition~\ref{prop:manyfirmsknownP} shows that the main comparative statics in Theorem~\ref{thm:mainknownP} are more likely to occur with more firms. This is despite equilibrium prices decreasing in the number of firms, which leaves less scope for cutting price in response to a greater switching cost. 

The next section discusses extensions and modifications of the main model studied above. The results remain robust when the firms can distinguish the incoming switchers from the initial customers.

\section{Extensions and discussion}
\label{sec:extensions}

\textbf{Price discrimination.} 
If the firms could price discriminate between consumers initially arriving at them and those switching in, then the results before the main theorem remain similar. Each firm sets two prices, so the FOC~(\ref{focknownP}) becomes two equations, one composed of the first and third line, which determines the firm's price for its initial consumers, and another composed of the second and fourth line, determining the price for switchers. The price for switchers is lower if $f_i'\geq0$, because $\inswitchvcutknownP \geq P_i +\max\set{0,\max_{j\neq i}\set{v_{j}-P_{j}^*-s}}$. This is in contrast to the case where the switchers do not know the price before switching---then, the firm would hold up the switchers with a high price. 

The SOC~(\ref{focknownP}) and subsequent equations before~(\ref{dfocdsoligopolyknownP}) are similarly split in two: one for the initial customers and another for the switchers. Each equation has the same sign as before the split, so the conclusions are the same. 

If the switching cost $s$ rises and $(P_i-c_i)\frac{\partial f_i(P_i+w)}{\partial P_i}\geq -f_i(P_i+w)$, then the price increases for the initial customers and decreases for the incoming switchers, as is apparent from~(\ref{dfocdsoligopolyknownP}). The intuition is that more can be extracted from the initial customers when they face greater difficulty leaving. The switchers, by contrast, need to be attracted with a lower price to overcome their cost of changing firms.

\textbf{Discussion.}
Higher switching costs reduce prices even in a one-shot one-product market without initial advertising by the sellers to attract consumers to direct their search. Sufficient for a costlier change of seller to be pro-competitive is an increasing convex density of consumer valuations, but this is not necessary. The result holds for any positive level of switching cost at which some consumers still switch.

The results are the more surprising because low-valuation consumers exit the market in response to a greater switching cost. This actually turns out to be the mechanism for the reduction in prices. Exit decreases the mass of inframarginal consumers, but has less of an effect on the marginal buyers and may even increase them. The ratio of inframarginal to marginal consumers falls, so the best response of each seller decreases. 

The more sellers there are in the market, the easier these comparative statics are to obtain, because with more firms, more inframarginal consumers arrive from rival sellers. These incoming switchers are the ones a firm would lose at a greater switching cost, thus the motive to cut price.

\bibliographystyle{ecta}
\bibliography{teooriaPaberid} 

\begin{thebibliography}{15}
\newcommand{\enquote}[1]{``#1''}
\expandafter\ifx\csname natexlab\endcsname\relax\def\natexlab#1{#1}\fi

\bibitem[\protect\citeauthoryear{Armstrong and Zhou}{Armstrong and
  Zhou}{2011}]{armstrong+zhou2011}
\textsc{Armstrong, M. and J.~Zhou} (2011): \enquote{Paying for prominence,}
  \emph{The Economic Journal}, 121, 368--395.

\bibitem[\protect\citeauthoryear{Cabral}{Cabral}{2009}]{cabral2009}
\textsc{Cabral, L.} (2009): \enquote{Small switching costs lead to lower
  prices,} \emph{Journal of Marketing Research}, 46, 449--451.

\bibitem[\protect\citeauthoryear{Cabral}{Cabral}{2016}]{cabral2016}
---\hspace{-.1pt}---\hspace{-.1pt}--- (2016): \enquote{Dynamic pricing in
  customer markets with switching costs,} \emph{Review of Economic Dynamics},
  20, 43--62.

\bibitem[\protect\citeauthoryear{Choi, Dai, and Kim}{Choi
  et~al.}{2018}]{choi+2018}
\textsc{Choi, M., A.~Y. Dai, and K.~Kim} (2018): \enquote{Consumer search and
  price competition,} \emph{Econometrica}, 86, 1257--1281.

\bibitem[\protect\citeauthoryear{Diamond}{Diamond}{1971}]{diamond1971}
\textsc{Diamond, P.~A.} (1971): \enquote{A model of price adjustment,}
  \emph{Journal of Economic Theory}, 3, 156--168.

\bibitem[\protect\citeauthoryear{Dub{\'e}, Hitsch, and Rossi}{Dub{\'e}
  et~al.}{2009}]{dube+2009}
\textsc{Dub{\'e}, J.-P., G.~J. Hitsch, and P.~E. Rossi} (2009): \enquote{Do
  switching costs make markets less competitive?} \emph{Journal of Marketing
  research}, 46, 435--445.

\bibitem[\protect\citeauthoryear{Haan, Moraga-Gonz{\'a}lez, and
  Petrikait{\.e}}{Haan et~al.}{2018}]{haan+2018}
\textsc{Haan, M.~A., J.~L. Moraga-Gonz{\'a}lez, and V.~Petrikait{\.e}} (2018):
  \enquote{A model of directed consumer search,} \emph{International Journal of
  Industrial Organization}, 61, 223--255.

\bibitem[\protect\citeauthoryear{Heinsalu}{Heinsalu}{2020}]{heinsalu2020searchcost}
\textsc{Heinsalu, S.} (2020): \enquote{Greater search cost reduces prices,}
  Working paper.

\bibitem[\protect\citeauthoryear{Klemperer}{Klemperer}{1987}]{klemperer1987}
\textsc{Klemperer, P.} (1987): \enquote{Markets with consumer switching costs,}
  \emph{Quarterly Journal of Economics}, 102, 375--394.

\bibitem[\protect\citeauthoryear{Lal and Sarvary}{Lal and
  Sarvary}{1999}]{lal+sarvary1999}
\textsc{Lal, R. and M.~Sarvary} (1999): \enquote{When and how is the Internet
  likely to decrease price competition?} \emph{Marketing Science}, 18,
  485--503.

\bibitem[\protect\citeauthoryear{Mauring}{Mauring}{2021}]{mauring2021}
\textsc{Mauring, E.} (2021): \enquote{Search and Price Discrimination Online,}
  Discussion Paper DP15729.

\bibitem[\protect\citeauthoryear{Murry and Zhou}{Murry and
  Zhou}{2020}]{murry+zhou2020}
\textsc{Murry, C. and Y.~Zhou} (2020): \enquote{Consumer Search and Automobile
  Dealer Colocation,} \emph{Management Science}, 66, 1909--1934.

\bibitem[\protect\citeauthoryear{Schultz}{Schultz}{2005}]{schultz2005}
\textsc{Schultz, C.} (2005): \enquote{Transparency on the consumer side and
  tacit collusion,} \emph{European Economic Review}, 49, 279--297.

\bibitem[\protect\citeauthoryear{Vitorino}{Vitorino}{2012}]{vitorino2012}
\textsc{Vitorino, M.~A.} (2012): \enquote{Empirical Entry Games with
  Complementarities: An Application to the Shopping Center Industry,}
  \emph{Journal of Marketing Research}, 49, 175--191.

\bibitem[\protect\citeauthoryear{Wolinsky}{Wolinsky}{1986}]{wolinsky1986}
\textsc{Wolinsky, A.} (1986): \enquote{True monopolistic competition as a
  result of imperfect information,} \emph{The Quarterly Journal of Economics},
  101, 493--511.

\end{thebibliography}
\end{document}